\documentclass[11pt]{article}
\usepackage{amssymb}
\usepackage{fullpage}
\usepackage{amsfonts}
\usepackage{amsmath}
\usepackage{cancel}
\usepackage{color}
\usepackage{fancyhdr}
\usepackage{ mathrsfs }
\usepackage{latexsym}
\usepackage{amsmath,amssymb,amsthm}
\usepackage{epsfig}
\usepackage{dsfont}

\newtheorem{definition}{Definition}
\newtheorem{theorem}{Theorem}

\newtheorem{corollary}{Corollary}

\begin{document}

\title{Strategyproof Mechanisms for One-Dimensional Hybrid and Obnoxious Facility Location Models}

\date{}

\author{
Itai Feigenbaum
\thanks{IEOR Department, Columbia University, New York,
NY;
{\tt itai@ieor.columbia.edu}} \and
Jay Sethuraman
\thanks{IEOR Department, Columbia University, New York,
NY;
{\tt jay@ieor.columbia.edu}. Research supported by NSF grant
CMMI-0916453 and CMMI-1201045.}
}

\maketitle

\begin{abstract}
We consider a strategic variant of the facility location problem on a closed interval. There are $n$ agents spread on that interval, divided into two types: type $1$ agents, who wish for the facility to be as far from them as possible, and type $2$ agents, who wish for the facility to be as close to them as possible. Our goal is to maximize a form of aggregated social benefit. We consider two social benefit functions: the sum of agents' utilities and the minimal agent utility, respectively denoted as the maxisum and the egalitarian objectives. The strategic aspect of the problem is that the agents' types and locations are not known to us, but rather reported to us by the agents-- an agent might misreport his type and\textbackslash or location in an attempt to move the facility towards or away from his true location. We therefore require the facility-locating mechanism to be strategyproof, namely that reporting truthfully is a dominant strategy for each agent. As simply maximizing the social benefit is generally not strategyproof, our goal is to design strategyproof mechanisms with good approximation ratios.

In our paper, we begin by providing a best-possible $3$- approximate deterministic strategyproof mechanism, as well as a $\frac{23}{13}$- approximate randomized strategyproof mechanism, both for the maxisum objective. We then provide a characterization of all deterministic strategyproof mechanisms, for the case where only type $1$ agents are present (also known as the obnoxious facility problem). We use that characterization to prove a (best-possible) lower bound of $3$ for the maxisum objective, and to show that the approximation ratio is unbounded for the egalitarian objective. When allowing for randomization, we prove lower bounds of $\frac{2}{\sqrt{3}}$ and $\frac{3}{2}$ for the maxisum and egalitarian objectives respectively. All lower bounds hold even when only type $1$ agents are present. Finally, while restricting ourselves to agents of type $1$ only, we consider a generalized model that allows an agent to control more than one location. In this generalized model, we provide $3$- and $\frac{3}{2}$- approximate strategyproof mechanisms for the maxisum objective in the deterministic and randomized settings, respectively (for the randomized setting, we actually provide a family of such mechanisms).
\end{abstract}

\section{Introduction}
When a public facility is to be constructed, the general population has varying preferences regarding its location. For example, if the facility is a transit station, members of the population who rely on public transportation are likely to want it located near them for convenience, while those who own a car will want the station located far from them in order to avoid the added commotion and traffic. These preferences must be taken into account by local government when it decides on the construction site. However, there are many situations in which the government does not know these preferences reliably and\textbackslash or cannot deduce them easily. Certainly, the government does not know whether a person prefers the facility near or far. In addition, the government might not know each person's reference point, that is where the person want the facility near to or far from. This can be due to poor data (the government does not fully know the location of houses and who lives where), but can also be because different people have different kinds of reference points (for example, a person who spends most of his time at work might use his workplace as his main reference point rather than his home).

More specifically, consider the problem of locating a single facility on a street, that is a closed interval. There are $n$ agents, located in the interval, divided into two types: type $1$ agents, who wish for the facility to be as far away from them as possible, and type $2$ agents, who wish for the facility to be as close to them as possible. In particular, the utility of a type $1$ agent equals his distance from the facility, while the utility of a type $2$ agent equals the length of the interval minus his distance from the facility.\footnote{See section 2 for a short explanation regarding the utility of type $2$ agents.} A social planner wishes to locate the facility in a way that maximizes some aggregated measure of the agents' utilities. However, we are interested in a variant of the problem (first introduced in \cite{DBLP:journals/teco/ProcacciaT13}), in which the types and locations of the agents are not known to the planner, but rather are reported to the planner by the agents themselves. In that case, an agent might misreport its type and\textbackslash or location if doing so will cause the planner to place the facility at a site more desirable to that agent. Due to this strategic aspect, the planner cannot simply locate the facility at the optimal location with respect to the reports. Instead, we require the mechanism used by the planner to be {\em strategyproof}: truthful reporting is a dominant strategy for each agent. Subject to this requirement, the planner's goal is to optimize the social benefit, in terms of worst case approximation ratio. We consider two social benefit functions: the maxisum function, which is simply the sum of the agents' utilities, and the egalitarian function, which is the minimum agent utility.

The strategic facility location problem and its variations have received a lot of attention in the recent literature. The case of the unbounded interval with type $2$ agents alone was studied in \cite{DBLP:journals/teco/ProcacciaT13}, \cite{DBLP:dblp_conf/sigecom/FeldmanW13} and \cite{DBLP:journals/corr/abs-1305-2446}; a notable characterization of deterministic strategyproof mechanisms in this setting is given in \cite{RePEc:kap:pubcho:v:35:y:1980:i:4:p:437-455}. The case of the bounded interval with agents of type $1$ alone, called the {\em obnoxious facility location} problem, was introduced in \cite{DBLP:journals/tcs/ChengYZ13} and further explored in \cite{DBLP:conf/cocoa/IbaraN12}. There is much related research, considering different graph topologies, different number of facilities, and more: see, for example, \cite{DBLP:dblp_conf/tamc/ChengHYZ13}, \cite{DBLP:journals/mor/AlonFPT10}, and \cite{DBLP:journals/corr/abs-0907-2049}. To the best of our knowledge, our paper is the first to consider the generalized, hybrid model which contains both types of agents.

Our main findings are summarized below:
\begin{itemize}
\item We design a $3$- approximate deterministic strategyproof mechanism, and a $\frac{23}{13}$- approximate randomized strategyproof mechanism for the maxisum objective.
\item We characterize deterministic strategyproof mechanisms when only type 1 agents are present.
\item We prove a lower bound of $3$ on the approximation ratio of deterministic strategyproof mechanisms for the maxisum objective, thus proving the optimality of the mechanism we provide for this setting. We also show that no deterministic strategyproof mechanism can provide a bounded approximation ratio for the egalitarian objective. These bounds hold even when all agents are of type $1$.
\item We prove lower bounds of $\frac{2}{\sqrt{3}}$ and $\frac{3}{2}$ on the approximation ratio of randomized strategyproof mechanisms for the maxisum and egalitarian objectives, respectively. These bounds hold even when all agents are of type $1$.
\item We consider a generalized model that allows an agent to control more than one location. In this model, we provide a $3$- and $\frac{3}{2}$- approximate strategyproof mechanisms for the maxisum objective in the deterministic and randomized settings respectively, assuming only type 1 agents are present (in the randomized setting, we actually provide a {\it family} of such mechanisms).
\end{itemize}

\section{Model}

Let $N = \{1, 2, \ldots, n\}$ be the set of agents, and let $I$ be the closed interval. We assume, without loss of generality, that $I=[0,2]$. Let the set of possible types be $T$ (generally, we will have $T=\{1,2\}$, but for some results we would like $T=\{1\}$). Each agent $i \in N$ reports a type $\theta_i \in T$ and a location $x_i \in I$. The vector ${\bf \theta}=(\theta_1, \theta_2, \ldots, \theta_n)$ is a type profile; for any $\alpha \in T$, we also use the notation $(\alpha, {\bf \theta_{-i}})=(\theta_1,\theta_2,\ldots,\theta_{i-1},\alpha,\theta_{i+1},\ldots,\theta_n)$, where ${\bf \theta_{-i}}=(\theta_1,\theta_2,\theta_{i-1},\theta_{i+1},\ldots,\theta_n)$ is a partial type profile of all agents but $i$. Similarly, the vector ${\bf x}=(x_1,\ldots,x_n)$ is a location profile (the notation $(\alpha,{\bf x_{-i}})$ is defined as $(x_1,\ldots,x_{i-1},\alpha,x_{i+1},x_n)$). A {\em deterministic} mechanism is a collection of functions $f = \{f_n | n \in \mathbb{N}\}$ such that each $f_n:T^n \times I^n \rightarrow I$ maps each type and location profiles ${\bf \theta}=(\theta_1,\ldots,\theta_n)$ and ${\bf x} = (x_1, x_2, \ldots, x_n)$ to the location of the facility. We 
use $f({\bf \theta},{\bf x})$ instead of $f_n({\bf \theta},{\bf x})$ when $n$ is clear from the context.
Similarly, a {\em randomized} mechanism is a collection of functions $f$ that maps each pair of type and location profiles to a probability distribution over $I$: if $f({\bf \theta},{\bf x})$ is the distribution $\pi$, then the facility is located by drawing a single sample from $\pi$.

We study deterministic and randomized mechanisms for
the problem of locating a single facility when the type and location of an agent are
{\em private} information to that agent and cannot be observed or otherwise
verified. It is therefore critical that the mechanism 
be {\em strategyproof}---it should be optimal for each agent $i$ to report
his {\em true} type $\theta_i$ and location $x_i$. To make this precise, we 
assume that if the facility is located at $y$, an agent's utility, equivalently benefit, is either $B(\theta_i,x_i,y)=|x_i-y|$ if $\theta_i=1$, or $B(\theta_i,x_i,y)=2-|x_i-y|$, if $\theta_i=2$ (the number $2$ is chosen merely because it is the length of $I$. Of course, what we really want is for the utility of a type $2$ agent to be $-|x_i-y|$. But in order to meaningfully discuss approximation ratios, we require nonnegative utilities. With this choice, the utility of each agent is between $0$ and $2$ regardless of his type, and the utility changes linearly with the distance the same way it would without adding the constant. This seems to be the natural choice). If the location 
of the facility is randomly distributed with distribution $\pi$, then the benefit of agent $i$
is simply $\mathbb{E}_{Y \sim \pi} [B(\theta_i,x_i,Y)]$,
where $Y$ is a random variable
with distribution $\pi$. The formal definition of strategyproofness is now:

\begin{definition}
A deterministic (randomized) mechanism $f$ is strategyproof if for each $i \in N$, 
each $\theta_i,\theta'_i \in T$, $x_i, x'_i \in I$ and for each
${\bf \theta_{-i}} \in T^{n-1}$, ${\bf x_{-i}} \in I^{n-1}$ we have $B(\theta_i,x_i, f((\theta_i, {\bf \theta_{-i}}),(x_i, {\bf x_{-i}}))) \geq B(\theta_i,x_i, f((\theta'_i, {\bf \theta_{-i}}),(x'_i, {\bf x_{-i}})))$ ($E_{Y \sim f((\theta_i, {\bf \theta_{-i}}),(x_i, {\bf x_{-i}}))} [B(\theta_i,x_i, Y)] \geq E_{Y \sim f((\theta'_i, {\bf \theta_{-i}}),(x'_i, {\bf x_{-i}}))}[B(\theta_i,x_i, Y)]$).
\end{definition}

In this paper we assume that locating a facility at $y$
when the type profile is ${\bf \theta}$ and the location profile is ${\bf x}$ gives the
{\em social benefit} $sb({\bf \theta},{\bf x},y)$, where we consider two possible options for $sb$: {\em maxisum}, defined by $sb({\bf \theta},{\bf x},y)=\sum_{i=1}^n B(\theta_i,x_i,y)$, and {\em egalitarian}, $sb({\bf \theta},{\bf x},y)=\min_{i \in N}{B(\theta_i,x_i,y)}$. When the facility is located according to a probability distribution $\pi$, {\em maxisum} is defined as $sb({\bf \theta},{\bf x},\pi)=E_{Y \sim \pi}[\sum_{i=1}^n B(\theta_i,x_i,Y)]$, and {\em egalitarian} as $sb({\bf \theta},{\bf x},\pi)=E_{Y \sim \pi}[\min_{i \in N}{B(\theta_i,x_i,Y)}]$.
The goal is to find a strategyproof
mechanism that does well with respect to maximizing (either definition of) the social  benefit. A natural
mechanism is the ``optimal'' mechanism: each pair of type profile ${\bf \theta}$ and location profile
${\bf x}$ is mapped to $OPT({\bf \theta},{\bf x})$, defined as\footnote{
If the social benefit is maximized by multiple locations $y$, an
exogenous tie-breaking rule is used to select one of the optimal locations.}
$OPT({\bf \theta},{\bf x}) \; \in \; \arg \max_{y \in I} sb({\bf \theta},{\bf x},y)$. However, the optimal mechanism is not generally strategyproof. For example, consider the case with two type $1$ agents, the first located at $\frac{2}{3}$ and the second located at $\frac{3}{2}$. In this case, the optimal mechanism for the maxisum objective would locate the facility at $0$. However, when the first agent is located at $\frac{1}{3}$ and the second agent is located at $\frac{3}{2}$, the optimal mechanism locates the facility at $2$. Since when the first agent's true location is $\frac{2}{3}$ he prefers the facility to be located at $2$ over $0$, he has an incentive to misreport his location to be $\frac{1}{3}$ instead, violating strategyproofness.

Given that strategyproofness and optimality cannot be achieved simultaneously,
it is necessary to find a tradeoff. In this paper we shall restrict ourselves
to strategyproof mechanisms that approximate the optimal social benefit as best
as possible: an $\alpha$- approximation ($\alpha \in [1,\infty)$) algorithm guarantees at least a $\frac{1}{\alpha}$ fraction
 of the optimal social benefit for every instance of the problem. Formally, the
approximation ratio of an algorithm $A$ is $\sup_{Q} \{ OPT(Q) / A(Q) \},$
where the supremum is taken over all possible instances $Q$ of the problem;
and $A(Q)$ and $OPT(Q)$ are, respectively, the benefits obtained by algorithm $A$
and the optimal algorithm on the instance $Q$. Our goal is to design a strategyproof
mechanism whose approximation ratio is as close to 1 as possible.

\section{Deterministic and Randomized Mechanisms for the Hybrid Model}

In this section, we provide a best-possible $3$- approximate deterministic strategyproof mechanism for the maxisum objective, as well as a $\frac{23}{13}$- approximate randomized strategyproof mechanism for the same objective.

\begin{theorem}
Let $R=\{i: \theta_i=1,x_i \leq 1\} \cup \{i: \theta_i=2,x_i \geq 1\}$ and $L=\{i:\theta_i=1,x_i>1\} \cup \{i: \theta_i=2,x_i<1\}$. Let $f$ be the mechanism that locates the facility at $2$ if $|R| \geq |L|$ and at $0$ otherwise. Then $f$ is a $3$- approximate strategyproof mechanism for the maxisum objective.
\end{theorem}

\begin{proof}
Strategyproofness is easy. Note that $R$ is the set of agents who weakly prefer the facility located at $2$ over $0$, and $L$ is the set of remaining agents. Since there are only two possible facility locations in this mechanism, the only case which requires analysis is when agent $i$ prefers the facility to be located at the endpoint {\it not} chosen by the mechanism. Assume without loss of generality that the facility is located at $2$, yet agent $i$ prefers the facility to be located at $0$ (that is, $i \in L$). Then, by misreporting, he cannot decrease $|R|$ and cannot increase $|L|$, and therefore regardless of his report, the facility will be located at $2$.\\

For the approximation ratio, let ${\bf \theta}$ and ${\bf x}$ be a type and location profile respectively. We would like to show that $\frac{sb({\bf \theta},{\bf x},a)}{sb({\bf \theta},{\bf x},f({\bf \theta},{\bf x}))}\leq 3$ for every possible facility location $a \in I$ \footnote{Note that the statement ``$\frac{sb({\bf \theta},{\bf x},a)}{sb({\bf \theta},{\bf x},f({\bf \theta},{\bf x}))}\leq 3$ for every possible facility location $a \in I$'' is equivalent to $\frac{sb({\bf \theta},{\bf x},OPT({\bf \theta},{\bf x}))}{sb({\bf \theta},{\bf x},f({\bf \theta},{\bf x}))}\leq 3$, as $OPT({\bf \theta},{\bf x})$ maximizes the numerator by definition and hence the ratio. However, we choose to analyze an arbitrary fixed $a$ rather than $OPT({\bf \theta},{\bf x})$ to avoid having to consider the impact agents' reports have on the optimal location of the facility.}. We will prove this for the case of $f({\bf \theta},{\bf x})=2$; the other case is similar. Let $R_j$ be the set of agents of type $j$ in $R$, and similarly let $L_j$ be the set of agents of type $j$ in $L$. Then $q_1=sb({\bf \theta},{\bf x},a)=\sum_{i \in R_1 \cup L_1} |x_i-a|+\sum_{i \in R_2 \cup L_2} (2-|x_i-a|)$, and $q_2=sb({\bf \theta},{\bf x},f({\bf \theta},{\bf x}))=\sum_{i \in R_1 \cup L_1}(2-x_i)+\sum_{i \in R_2 \cup L_2}x_i$. If the ratio $\frac{q_1}{q_2} <1$, there is nothing to prove. Assume $\frac{q_1}{q_2} \geq 1$. Note that for $i \in R_1 \cup L_1$, increasing $x_i$ by $\alpha$ decreases the denominator by $\alpha$, and decreases the numerator by at most $\alpha$ (might even increase the numerator in some cases). Since the ratio is at least $1$, such a change increases the ratio. Similarly, for $i \in R_2 \cup L_2$, decreasing $x_i$ has the same effect. Given that $x_i \leq 1$ for $i \in R_1$, $x_i \geq 1$ for $i \in R_2$, and $x_i \in [0,2]$ for all $i$, this implies $\frac{q_1}{q_2} \leq \frac{\sum_{i \in R_1}|a-1|+\sum_{i \in L_1 \cup L_2}(2-a)+\sum_{i \in R_2}(2-|a-1|)}{|R|}$. We break our proof into two cases:
\begin{enumerate}
\item $a \in [0,1]$: $\sum_{i \in R_1}|a-1|+\sum_{i \in L_1 \cup L_2}(2-a)+\sum_{i \in R_2}(2-|a-1|) = |R_1|+2|L_1|+2|L_2|+|R_2|+a(-|R_1|-|L_1|-|L_2|+|R_2|) \leq \max{\{|R|+2|L|,|L|+2|R_2|\}}$, where the inequality follows from the fact that the maximum is obtained when $a \in \{0,1\}$ (If $-|R_1|-|L_1|-|L_2|+|R_2| \leq 0$, then it is obtained at $a=0$, and otherwise at $a=1$). Note that $|L| \leq |R|$ since $f({\bf x})=2$, and that $|R_2| \leq |R|$ by definition. Thus, $\max{\{|R|+2|L|,|L|+2|R_2|\}} \leq 3|R|$. Therefore, $\frac{q_1}{q_2} \leq \frac{3|R|}{|R|}=3$
\item $a \in [1,2]$: $\sum_{i \in R_1}|a-1|+\sum_{i \in L_1 \cup L_2}(2-a)+\sum_{i \in R_2}(2-|a-1|) = -|R_1|+2|L_1|+2|L_2|+3|R_2|+a(|R_1|-|L_1|-|L_2|-|R_2|) \leq \max{\{|L|+2|R_2|,|R|\}}$. Again, both terms we're maximizing over are no more than $3|R|$, and so again $\frac{q_1}{q_2} \leq 3$.
\end{enumerate}
\end{proof}

Later, we prove a lower bound of $3$ on the approximation ratio possible under strategyproofness. Thus, the approximation ratio achieved by this mechanism is best-possible. Moreover, in the obnoxious facility model (when no type $2$ agents exist), the above mechanism reduces to the deterministic mechanism proposed in \cite{DBLP:journals/tcs/ChengYZ13}, who proved that it is a $3$- approximation for that special case.

We now use randomization in an attempt to lower the approximation ratio. Getting a $2$- approximation is easy: choosing each endpoint with probability $\frac{1}{2}$ is a $2$- approximate strategyproof mechanism \footnote{In \cite{DBLP:journals/tcs/ChengYZ13}, the authors note that this mechanism is $2$- approximate for the obnoxious facility model; this still holds true for the hybrid model.}. However, we can do better:

\begin{theorem}
Let $p_1=\frac{12}{23}$, $p_2=\frac{8}{23}$, and $p_3=\frac{3}{23}$. 
Consider the following randomized mechanism $f$. 
If $|R| \geq |L|$, then $P(f({\bf \theta},{\bf x})=2)=p_1$ and $P(f({\bf \theta},{\bf x})=0)=p_2$; if $|R|<|L|$, then $P(f({\bf \theta},{\bf x})=2)=p_2$ and $P(f({\bf \theta},{\bf x})=0)=p_1$; and either way, $P(f({\bf \theta},{\bf x})=1)=p_3$. The mechanism $f$ is a strategyproof, $\frac{23}{13}$- approximate mechanism.
\end{theorem}

\begin{proof}
Strategyproofness is proved similarly to Theorem 3.1. For the approximation ratio, we would like to show that  $\frac{sb({\bf \theta},{\bf x},a)}{sb({\bf \theta},{\bf x},f({\bf \theta},{\bf x}))}\leq \frac{23}{13}$ for every possible facility location $a \in I$. We will prove this for the case of $|R| \geq |L|$; the other case is similar. Define $R_j$ and $L_j$ as in the proof of Theorem 3.1. We begin by noting that the approximation ratio is bounded from above by $\frac{46}{19}$ (it is easy to see that every agent is guaranteed a benefit of at least $\frac{19}{23}$ in our mechanism, while the maximal benefit of any agent is $2$). Note that the mechanism's expected benefit is $q_2=-\frac{7}{23}\sum_{i \in R_1}x_i-\frac{1}{23}\sum_{i \in L_1}x_i+\frac{1}{23}\sum_{i \in R_2}x_i+\frac{7}{23}\sum_{i \in L_2}x_i+\frac{27}{23}|R_1|+\frac{21}{23}|L_1|+\frac{25}{23}|R_2|+\frac{19}{23}|L_2|$. We break into cases:
\begin{enumerate}
\item $a \in [0,1]$: in this case, the benefit from locating the facility at $a$ is $q_1=\sum_{i \in R_1} |a-x_i|+\sum_{i \in L_1}(x_i-a)+\sum_{i \in R_2}(2+a-x_i)+\sum_{i \in L_2}(2-|a-x_i|)$. Note that the ratio $\frac{q_1}{q_2}$ increases with $x_i$ for $i \in L_1$, and decreases with $x_i$ for $i \in R_2$; thus, to maximize it, we set $x_i=2$ for $i \in L_1$ and $x_i=1$ for $i \in R_2$. For $i \in L_2$, $x_i=a$ maximizes the ratio (note that changing $x_i$ by $\alpha$ decreases the numerator by $\alpha$ and either increases or decreases the denominator by $\frac{7}{23}\alpha$, a change that decreases the ratio since it is known to be no more than $\frac{46}{19} < \frac{23}{7}$). Depending on $a$, to maximize the ratio we need to set $x_i=0$ for all $i \in R_1$ or $x_i=1$ for all $i \in R_1$. We check both cases:
\begin{enumerate}
\item $x_i=0$ for all $i \in R_1$. Then the ratio becomes $\frac{a(|R_1|-|L_1|+|R_2|)+2|L_1|+|R_2|+2|L_2|}{\frac{27}{23}|R_1|+\frac{19}{23}|L_1|+\frac{26}{23}|R_2|+\frac{19}{23}|L_2|+\frac{7}{23}a|L_2|}$. As $|R_1|+|R_2|-|L_1| \geq |L_2| \geq 0$, this ratio increases with $a$ and hence maximized at $a=1$ \footnote{Of course, for agents $i \in L_2$ we technically cannot have $x_i=1$, but the bound holds nonetheless.}, which leads to the ratio: $\frac{|R_1|+|L_1|+2|R_2|+2|L_2|}{\frac{27}{23}|R_1|+\frac{19}{23}|L_1|+\frac{26}{23}|R_2|+\frac{26}{23}|L_2|} \leq \frac{23}{13}$.
\item $x_i=1$ for all $i \in R_1$. The ratio becomes $\frac{a(-|R_1|-|L_1|+|R_2|)+|R_1|+2|L_1|+|R_2|+2|L_2|}{\frac{20}{23}|R_1|+\frac{19}{23}|L_1|+\frac{26}{23}|R_2|+\frac{19}{23}|L_2|+\frac{7}{23}a|L_2|}$. Maximization occurs either at $a=1$ or at $a=0$. We check both cases:
\begin{enumerate}
\item $a=1$: the ratio becomes $\frac{|L_1|+2|R_2|+2|L_2|}{\frac{20}{23}|R_1|+\frac{19}{23}|L_1|+\frac{26}{23}|R_2|+\frac{26}{23}|L_2|} \leq \frac{23}{13}$.
\item $a=0$: the ratio becomes $\frac{|R_1|+2|L_1|+|R_2|+2|L_2|}{\frac{20}{23}|R_1|+\frac{19}{23}|L_1|+\frac{26}{23}|R_2|+\frac{19}{23}|L_2|}$. As $|L| \leq |R|$, it follows that the ratio is bounded from above by $\frac{1+2}{\frac{20}{23}+\frac{19}{23}}=\frac{23}{13}$.
\end{enumerate}
\end{enumerate}
\item $a \in [1,2]$: in this case, the benefit from locating the facility at $a$ is $q_1=\sum_{i \in R_1} (a-x_i)+\sum_{i \in L_1}|a-x_i|+\sum_{i \in R_2}(2-|a-x_i|)+\sum_{i \in L_2}(2+x_i-a)$. Similarly to the analysis in the previous case, we get that the ratio $\frac{q_1}{q_2}$ is maximized when  $x_i=0$ for $i \in R_1$, $x_i=1$ for $i \in L_2$, $x_i=a$ for $i \in R_2$, and $x_i=1$ for all $i \in L_1$ or $x_i=2$ for all $i \in L_1$. We break into cases:
\begin{enumerate}
\item $x_i=1$ for all $i \in L_1$: the ratio becomes $\frac{a(|R_1|-|L_2|+|L_1|)-|L_1|+2|R_2|+3|L_2|}{\frac{1}{23}a|R_2|+\frac{27}{23}|R_1|+\frac{20}{23}|L_1|+\frac{25}{23}|R_2|+\frac{26}{23}|L_2|}$. Maximum is obtained at either $a=1$ or $a=2$:
\begin{enumerate}
\item $a=1$: the ratio becomes $\frac{|R_1|+2|R_2|+2|L_2|}{\frac{27}{23}|R_1|+\frac{20}{23}|L_1|+\frac{26}{23}|R_2|+\frac{26}{23}|L_2|} \leq \frac{23}{13}$.
\item $a=2$: the ratio becomes $\frac{2|R_1|+|L_1|+2|R_2|+|L_2|}{\frac{27}{23}|R_1|+\frac{20}{23}|L_1|+\frac{27}{23}|R_2|+\frac{26}{23}|L_2|} \leq \frac{46}{27}<\frac{23}{13}$.
\end{enumerate}
\item $x_i=2$ for all $i \in L_1$: the ratio becomes $\frac{a(|R_1|-|L_2|-|L_1|)+2|L_1|+2|R_2|+3|L_2|}{\frac{1}{23}a|R_2|+\frac{27}{23}|R_1|+\frac{19}{23}|L_1|+\frac{25}{23}|R_2|+\frac{26}{23}|L_2|}$. Maximum is obtained at either $a=1$ or $a=2$:
\begin{enumerate}
\item $a=1$: the ratio becomes $\frac{|R_1|+|L_1|+2|R_2|+2|L_2|}{\frac{27}{23}|R_1|+\frac{19}{23}|L_1|+\frac{26}{23}|R_2|+\frac{26}{23}|L_2|} \leq \frac{23}{13}$.
\item $a=2$: the ratio becomes $\frac{2|R_1|+2|R_2|+|L_2|}{\frac{27}{23}|R_1|+\frac{19}{23}|L_1|+\frac{27}{23}|R_2|+\frac{26}{23}|L_2|} \leq \frac{46}{27}<\frac{23}{13}$.
\end{enumerate}
\end{enumerate}
\end{enumerate}
\end{proof}

The approximation ratio of the mechanism above is tight: when there are two agents of different types, with the type $1$ agent at $1$ and the type $2$ agent at $0$, the optimal benefit is $3$, whereas the mechanism's expected benefit is $\frac{39}{23}$, and the ratio is exactly $\frac{23}{13}$.

\section{Characterization of Deterministic Mechanisms for the Obnoxious Facility Model}

We now focus on the special case where there are no type 2 agents, namely $T=\{1\}$, also called the obnoxious facility model. The assumption that there are no type 2 agents will remain in effect for the rest of this section. Note that in this case an agent's report of its type is meaningless, and so we drop it from the input of the mechanism. In this section, we characterize all deterministic strategyproof mechanisms for the obnoxious facility model. Similar results have been independently obtained by others \cite{han2012moneyless, BBM, VM}. We begin with a temporary, somewhat weak characterization of deterministic mechanisms, in terms of single agent deviations:

\begin{theorem}[Reflection Theorem]
For any deterministic mechanism $f$, agent $i \in N$, and partial location profile ${\bf x_{-i}}$, define $f_{\bf x_{-i}}(a)=f(a,{\bf x_{-i}})$ \footnote{Note that when $i \neq j$, ${\bf x_{-i}}$ and ${\bf x_{-j}}$ are distinct objects, regardless of the values of their coordinates.}. Then, the mechanism $f$ is strategyproof iff each $f_{\bf x_{-i}}$ is of the following form: there exists (not necessarily distinct) $\alpha_{\bf x_{-i}}, \beta_{\bf x_{-i}} \in I$, such that $\beta_{\bf x_{-i}} \geq \alpha_{\bf x_{-i}}$ and:
\begin{enumerate}
\item $f_{\bf x_{-i}}(a)=\beta_{\bf x_{-i}}$ for $0 \leq a <\frac{\alpha_{\bf x_{-i}}+\beta_{\bf x_{-i}}}{2}$
\item $f_{\bf x_{-i}}(a)=\alpha_{\bf x_{-i}}$ for $\frac{\alpha_{\bf x_{-i}}+\beta_{\bf x_{-i}}}{2} <a \leq 2$
\item $f_{\bf x_{-i}}(\frac{\alpha_{\bf x_{-i}}+\beta_{\bf x_{-i}}}{2}) \in \{\alpha_{\bf x_{-i}},\beta_{\bf x_{-i}}\}$
\end{enumerate}
If $\alpha_{\bf x_{-i}} \neq \beta_{\bf x_{-i}}$, we call $\frac{\alpha_{\bf x_{-i}}+\beta_{\bf x_{-i}}}{2}$ the reflection point of $i$ for the partial profile ${\bf x_{-i}}$.
\end{theorem}

\begin{proof}
First, assume that $f$ is of the form described above. On a partial location profile ${\bf x_{-i}}$, agent $i$ can only get the mechanism to choose one of (up to) two locations: $\alpha_{\bf x_{-i}}$ or $\beta_{\bf x_{-i}}$.  Let $q=\frac{\alpha_{\bf x_{-i}}+\beta_{\bf x_{-i}}}{2}$. If $x_i=q$, his distance from the two locations is equal, and so he is indifferent between them. If $x_i \in [0,q)$, $\beta_{\bf x_{-i}}$ is weakly \footnote{Weakly because it is possible that $\alpha_{\bf x_{-i}}=\beta_{\bf x_{-i}}$.} farther from him than $\alpha_{\bf x_{-i}}$, and so he weakly prefers $\beta_{\bf x_{-i}}$, which is what the mechanism chooses, so he has no incentive to deviate. The case of $x_i \in (q,2]$ is similar. Thus, $f$ is strategyproof.\\

On the other hand, assume that $f$ is strategyproof. Fix a location profile ${\bf x}$ and an agent $i$. Let $g=f_{\bf x_{-i}}$ and let $\beta=g(0)$. Let $S=\{a \in I: g(a) \neq \beta\}$. If $S$ is empty, then $g$ is constant, and we're done. So assume $S$ is nonempty. Consider $m=\inf{S}$. Note that if $\beta<m$, then an agent located at $\beta$ can benefit from a deviation to any point in $S$. Thus, $\beta \geq m$. Let $\alpha=2m-\beta$ (note that with our knowledge at this point, it might be the case that $\alpha$ is negative and hence not in $I$; our proof is careful not to assume otherwise). We begin by claiming that either $g(m)=\alpha$, or that $m$ is a limit point of the set $K=\{a \in I:g(a)=\alpha\}$. There are two cases to consider:
\begin{enumerate}
\item $m \in S$. Note that in this case $m>0$. We claim that in this case $g(m)=\alpha$. Assume otherwise, namely $g(m)=\alpha' \neq \alpha$. Note also that $\alpha' \neq \beta$ (since $m \in S$), and thus $m-\alpha \neq |m-\alpha'|$. There are two subcases:
\begin{enumerate}
\item $m-\alpha<|m-\alpha'|$. In that case, note that as long as the agent is to the left of $m$, the facility is located at $\beta$. Thus, if the agent is located at $m-\epsilon$ for some $\epsilon>0$, his distance from the facility is $\beta-m+\epsilon$. His distance from $\alpha'$ is at least $|m-\alpha'|-\epsilon$. However, as $\beta-m=m-\alpha<|m-\alpha'|$, we may choose $\epsilon$ small enough so that $|m-\alpha'|-\epsilon>\beta-m+\epsilon$. In this case, the agent's deviation from $m-\epsilon$ to $m$ is beneficial to that agent.
\item $m-\alpha>|m-\alpha'|$. In this case, it is still true that as long as the agent is to the left of $m$, the facility is located at $\beta$. As the distance of $\beta-m=m-\alpha>|m-\alpha'|$, it follows that an agent located at $m$ will benefit from deviating to the left.
\end{enumerate}
So indeed, $g(m)=\alpha$.
\item $m \notin S$. Then $m$ is a limit point of $S$ (by definition). In this case we claim that $m$ is a limit point of $K$. Furthermore, note that $\beta>m$, since if $\beta=m$, then as $m \notin S$, $g(\beta)=\beta$, and the agent can benefit by deviating from $\beta$ to any point in $S$. We note that since when the agent is located at $m$, the facility is located at $\beta$, strategyproofness dictates that the facility is always located in $[\alpha,\beta]$, no matter where the agent reports his location to be. Now, assume $m$ is not a limit point of $K$. Thus, it follows it must be a limit point of either $K_1=\{a \in I: \alpha< g(a) \leq m\}$ or $K_2=\{a \in I: m \leq g(a) < \beta\}$ \footnote{Note that since $g(a) \in [\alpha,\beta]$ for all $a \in I$, and $m$ is a limit point of $S$, it follows that $m$ is a limit point of $\{a \in I:a \in [\alpha,\beta)\}=K \cup K_1 \cup K_2$.}. So, there are two cases:
\begin{enumerate}
\item $m$ is a limit point of $K_1$. In particular, there exists some $\epsilon>0$ such that $m+\epsilon \in K_1$. We consider the following subcases:
\begin{enumerate}
\item There exists $0<\epsilon'<\epsilon$ s.t. $m+\epsilon' \in K_1$ and $g(m+\epsilon')<g(m+\epsilon)$; in this case, a deviation from $m+\epsilon$ to $m+\epsilon'$ is beneficial.
\item There exists $0<\epsilon'<\epsilon$ s.t. $m+\epsilon' \in K_1$ and $g(m+\epsilon')>g(m+\epsilon)$; in this case, a deviation from $m+\epsilon'$ to $m+\epsilon$ is beneficial.
\item $g(m+\epsilon')=g(m+\epsilon)$ for all $0<\epsilon'<\epsilon$ s.t. $m+\epsilon' \in K_1$. As $m$ is a limit point of $K_1$ and all points in $K_1$ are to the right of $m$, it follows that we may choose $\epsilon'$ as small as we want. For $\epsilon'$ small enough, this would imply that the deviation from $m+\epsilon'$ to $m$ is beneficial: the distance of $m+\epsilon'$ from $g(m+\epsilon)$ is $m-g(m+\epsilon)+\epsilon'$ and the distance of $m+\epsilon'$ from $\beta$ is $\beta-m-\epsilon'$. As $\beta-m>m-g(m+\epsilon)$, we may choose $\epsilon'$ small enough to make the deviation in question beneficial.
\end{enumerate}
\item $m$ is a limit point of $K_2$. So, there exists $0<\epsilon<\frac{\beta-m}{2}$ such that $m+\epsilon \in K_2$. The agent can benefit by deviating from $m+\epsilon$ to $m$ (since the facility will be sent from $g(m+\epsilon)$ to $\beta$, and since $m \leq g(m+\epsilon) <\beta$, $g(m+\epsilon)$ is closer to $m+\epsilon$ than $\beta$).
\end{enumerate}
Hence, by strategyproofness, we have reached a contradiction, and so $m$ must be a limit point of $K$.
\end{enumerate}

We have shown that if $m \in S$ then $g(m)=\alpha$, and otherwise by definition of $S$ $g(m)=\beta$. To complete the proof, we must show that $g(a)=\alpha$ for all $a>m$. Assume otherwise for some $a'>m$. First, note that since $g(m)$ is either $\alpha$ or $\beta$, $g(a) \in [\alpha,\beta]$ for all $a \in I$ by strategyproofness. Note that within this range, $\alpha$ is the point farthest from $a'$. Thus, the agent has an incentive to deviate from $a'$ to any point $a''$ for which $g(a'')=\alpha$, where the existance of such a point is guaranteed by the above discussion. This is a contradiction, and so we've completed our proof. \footnote{Note that we didn't actually need that $m$ is a limit point of $K$ in the second case, but merely that $K$ is nonempty. However, this doesn't seem to lead to a much simpler proof, so we stick with the more general argument.}
\end{proof}

As a corollary of the above theorem, we can deduce:

\begin{corollary}
For any deterministic strategyproof mechanism $f$, and any $n \in \mathbb{N}$, $R_n^f=\{f_n({\bf x}):{\bf x} \in I^n\}$ is finite.
\end{corollary}

\begin{proof}
Let ${\bf x}$ be an arbitrary profile, and set ${\bf x}^0={\bf x}$. For a given profile ${\bf x}^{i-1}$, consider the profiles ${\bf z}=(0,{\bf x_{-i}}^{i-1})$ and ${\bf z}'=(2,{\bf x_{-i}}^{i-1})$. By the reflection theorem, at least one of $f({\bf z})=f({\bf x}^{i-1})$ or $f({\bf z}')=f({\bf x}^{i-1})$ is true. This is trivial if agent $i$ has no reflection point at ${\bf x}^{i-1}_{\bf -i}$. Otherwise, if he has such a reflection point $m$, if $x_i^{i-1}>m$ or $x_i^{i-1}<m$, he may deviate to $2$ or $0$ respectively without changing the facility's location; if $x_i^{i-1}=m$, then still the reflection theorem gives that $f(m,{\bf x}^{i-1}_{\bf -i})$ equals one of $f({\bf z})$ or $f({\bf z}')$, as required \footnote{$f({\bf z})$ if $f(m,{\bf x}^{i-1}_{\bf -i})=\beta_{{\bf x_{-i}}}$, and $f({\bf z'})$ if $f(m,{\bf x}^{i-1}_{\bf -i})=\alpha_{{\bf x_{-i}}}$.}. Set ${\bf x}^i$ equal to a profile among ${\bf z}$ and ${\bf z}'$ satisfying the equality. Thus, ${\bf x}^n$ is a profile in which all agents are located at the endpoints and $f({\bf x}^n)=f({\bf x})$. Since ${\bf x}$ was arbitrary, we have that all elements of $R_n^f$ can be obtained by applying the mechanism to profiles locating all agents at the endpoints. Since there are only finitely many such profiles, $R_n^f$ is finite.
\end{proof}

Now it is time for our strong characterization result. Consider the following definition:

\begin{definition}
Let $f$ be a deterministic mechanism s.t. $|R_n^f| \leq 2$ for all $n \in \mathbb{N}$. For each $n \in \mathbb{N}$, let $R_n^f=\{\alpha_n,\beta_n\}$ s.t. $\beta_n \geq \alpha_n$ \footnote{$\alpha_n=\beta_n$ is possible.}, and let $m_n=\frac{\alpha_n+\beta_n}{2}$. For any $n \in \mathbb{N}$, for every profile ${\bf x} \in I^n$, consider the partition of the agents $L^{\bf x}=\{i \in N:x_i<m_n\}$, $M^{\bf x}=\{i \in N:x_i=m_n\}$, and $E^{\bf x}=\{i \in N:x_i>m_n\}$.  We say that $f$ is a {\emph midpoint} mechanism if it satisfies the following property: for any $n \in \mathbb{N}$, let ${\bf x},{\bf y} \in I^n$ be any profiles s.t. $f({\bf x})=\beta_n$ and $f({\bf y})=\alpha_n$. If $\beta_n>\alpha_n$, then there exists an agent $i$ which satisfies one of the following:
\begin{enumerate}
\addtolength{\itemindent}{0.7cm}
\item[(D-1)] $i \in L^{\bf x}$ and $i \in M^{\bf y}$
\item[(D-2)] $i \in L^{\bf x}$ and $i \in E^{\bf y}$
\item[(D-3)] $i \in M^{\bf x}$ and $i \in E^{\bf y}$
\end{enumerate}
\end{definition}
 
This definition is simple to interpret: the mechanism can switch the facility location from right to left or from left to right only when an agent crosses the midpoint in the opposite direction.

In \cite{DBLP:conf/cocoa/IbaraN12}, the authors show that for a strategyproof mechanism $f$, $|R_n^f| \leq 2$ whenever $R_n^f$ is a finite set\footnote{While they assume anonymity, the proof of this fact does not rely on that assumption.} \footnote{Note that this is not an immediate implication of the Gibbard-Satterthwaite theorem, since the agents, by reporting a location, cannot arbitrarily "rank" the locations in $R_n^f$; only some rankings are feasible.}. Using that, we can now show:

\begin{theorem}
A deterministic mechanism $f$ is strategyproof iff it is a midpoint mechanism.
\end{theorem}

\begin{proof}
First, consider a given midpoint mechanism $f$, and fix $n \in \mathbb{N}$. If $f_n$ is constant, then clearly it is strategyproof. Otherwise, $|R_n^f|=2$. Consider a profile ${\bf x} \in I^n$ and an agent $i \in N$. The facility can only be located at $\alpha_n$ or $\beta_n$. If $i \in M_n^{\bf x}$, he is indifferent between the two points, and thus has no incentive to deviate. If $i \in E_n^{\bf x}$, he prefers the facility to be located at $\alpha_n$; however, if $f({\bf x}) \neq \alpha_n$, note that agent $i$ cannot move the facility to $\alpha_n$ by deviating- the rest of the agents remain still, and he himself cannot be the agent required in the definition of the midpoint property (that agent cannot be in $E_n^{\bf x}$). The proof is similar for $i \in L_n^{\bf x}$.\\

Now, assume instead that $f$ is a strategyproof mechanism. Fix $n \in \mathbb{N}$. By corollary 4.2, $R_n^f$ is finite, and thus by Ibara's and Nagamochi's result, $|R_n^f| \leq 2$. If $R_n^f$ is a singleton there is nothing to prove; thus, assume $|R_n^f|=2$, and let $\alpha_n,\beta_n \in R_n^f$ s.t. $\beta_n>\alpha_n$. Let ${\bf x},{\bf y} \in I^n$ s.t. $f({\bf x})=\beta_n$ and $f({\bf y})=\alpha_n$. Consider the sequence of profiles ${\bf z}^i$, defined for $i=0,\ldots,n$ via ${\bf z}^i_j={\bf x}_j$ if $j > i$ and  ${\bf z}^i_j={\bf y}_j$ otherwise. Assume no agent satisfies at least one of (D-1), (D-2) and (D-3). Then, when agent $i$ deviates in ${\bf z}^{i-1}$ to create profile ${\bf z}^i$, he does not cross $m_n$ from left to right (i.e. moving from ${\bf z}^{i-1}_i<m_n$ to ${\bf z}^i_i \geq m_n$ or ${\bf z}^{i-1}_i \leq m_n$ to ${\bf z}^i_i > m_n$). As the possible facility locations are $\alpha_n$ and $\beta_n$, $m_n$ is his only candidate for reflection point in ${\bf z}^{i-1}_{\bf -i}$. Thus, the reflection theorem implies that he cannot change the facility location to $\alpha_n$ by deviating. Hence, $f({\bf y})=f({\bf z}^n)=f({\bf z}^0)=f({\bf x})$, contradiction.
\end{proof}

We note that Ibara and Nagmochi have characterized all anonymous mechanisms under the assumption that $R_n^f$ is finite for all $n \in \mathbb{N}$, using what they called ``valid threshold mechanisms". Our proofs easily translate to the anonymous case, and under anonymity, our midpoint mechanisms become equivalent to valid threshold mechanisms. Thus, our work allows the removal of the finite $R_n^f$ assumption for the anonymous case as well.

\section{Lower Bounds on Deterministic Mechanisms}

We can use our characterization to obtain lower bounds on the possible approximation ratios for the maxisum and egalitarian objectives in the deterministic setting. We note that negative results obtained when $T=\{1\}$ clearly hold when $T=\{1,2\}$.

\begin{theorem}
No deterministic strategyproof mechanism $f$ can provide an approximation ratio better than $3$ for the maxisum objective, even when $T=\{1\}$.
\end{theorem}

\begin{proof}
Let $f$ be a deterministic strategyproof mechanism. Assume $T=\{1\}$. Let $n \in \mathbb{N}$ be even. If $f_n$ is constant, the approximation ratio is clearly unbounded. If $R_n^f$ is not a singleton, then by Theorem 4.4, $|R_n^f|=2$. Consider the profile ${\bf x} \in I^n$ which locates agents $1$ through $\frac{n}{2}$ at $\alpha_n$, agents $\frac{n}{2}+1$ through $n$ at $\beta_n$ (where $\alpha_n<\beta_n$ are as in the definition of midpoint mechanism). Assume without loss of generality that $f({\bf x}) \neq \alpha_n$. Consider the profile ${\bf y}$ which locates agents $1$ through $\frac{n}{2}$ at $m_n-\epsilon$ for some $\epsilon>0$, and agrees with ${\bf x}$ on the rest of the agents. Since no deviating agent reaches $m_n$, the facility location doesn't change, that is $f({\bf y})=f({\bf x})=\beta_n$. Locating the facility at $\beta_n$ (on profile ${\bf y}$) leads to a benefit of $\frac{n}{2} \cdot (\frac{\beta_n-\alpha_n}{2}+\epsilon)$, while locating the facility at $\alpha_n$ leads to a benefit of $\frac{n}{2} \cdot (\frac{\beta_n-\alpha_n}{2}-\epsilon)+\frac{n}{2}(\beta_n-\alpha_n)$, and sending $\epsilon \rightarrow 0$ gives us the required result. \footnote{If $n$ is odd, we could still make this proof work by locating the additional agent at $m_n$ and send $n \rightarrow \infty$.}
\end{proof}

By Theorem 3.1, our lower bound is best-possible. Our characterization can also be used to get a lower bound for the egalitarian objective:

\begin{theorem}
No deterministic strategyproof mechanism $f$ can provide a bounded approximation ratio for the egalitarian objective, even when $T=\{1\}$.
\end{theorem}

\begin{proof}
Assume $T=\{1\}$. For any $n \geq 2$, $|R_n^f| \leq 2$ by Theorem 4.4. Consider any profile which locates at least one agent at each point in $R_n^f$; any such profile leads to a social benefit of $0$ for the mechanism, whereas the optimal benefit is positive.
\end{proof}

\section{Lower Bounds on Randomized Mechanisms}

We begin with the maxisum objective. We provide a lower bound of $\frac{2}{\sqrt{3}}$ on the approximation ratio of randomized strategyproof mechanisms.

\begin{theorem}
No randomized strategyproof mechanism can provide an approximation ratio better than $\frac{2}{\sqrt{3}}$ for the maxisum objective, even when $T=\{1\}$.
\end{theorem}

\begin{proof}
Let $f$ be a randomized strategyproof mechanism which provides an approximation ratio $c<\frac{2}{\sqrt{3}}$ for the maxisum objective. Consider the case where $N=\{1,2\}$ and $T=\{1\}$, and let $a=2\sqrt{3}-3$. Let ${\bf x}$ be the location profile in which $x_1=1-a$ and $x_2=1+a$. Assume without loss of generality that $P(f({\bf x})<x_1) \geq P(f({\bf x})>x_2)$. The expected distance of the facility from $x_1$ on this profile is at most $(1-a)P(f({\bf x})<x_1)+(1+a)P(f({\bf x})>x_2)+2aP(x_1 \leq f({\bf x}) \leq x_2)$; as $P(f({\bf x})<x_1) \geq P(f({\bf x})>x_2)$ and $2a \leq 1$, this implies that the expected distance of the facility from $x_1$ on profile ${\bf x}$ is at most $1$.\\

Let ${\bf y}$ be the profile in which $y_1=0$ and $y_2=1+a$. Let $b=E[f({\bf y})|f({\bf y})>y_2]-y_2$, and let $p=P(f({\bf x})>y_2)$. The mechanism's expected benefit is $1+a+2bp$, while the optimal cost is $3-a$. To maintain approximation ratio of $c$, we must have $1+a+2bp \geq \frac{3-a}{c}$, which implies $bp \geq \frac{3-a}{2c}-\frac{1+a}{2}$. Also, as $b \leq 1-a$, we have that $p \geq \frac{1}{1-a}(\frac{3-a}{2c}-\frac{1+a}{2})$. Now, the expected distance of the facility from $x_1$ on ${\bf y}$ is $(2a+b)p \geq (\frac{2a}{1-a}+1)(\frac{3-a}{2c}-\frac{1+a}{2}) > (\frac{2a}{1-a}+1)(\sqrt{3}\frac{3-a}{4}-\frac{1+a}{2})=1$. This violates strategyproofness, as agent 1 has an incentive to misreport his location to be $0$ when the location profile is ${\bf x}$.
\end{proof}

Next, we show a lower bound of $\frac{3}{2}$ for the egalitarian objective:

\begin{theorem}
No randomized strategyproof mechanism can provide an approximation ratio better than $\frac{3}{2}$ for the egalitarian objective, even when $T=\{1\}$.
\end{theorem}

\begin{proof}
Assume $T=\{1\}$. Let $f$ be such a mechanism, with approximation ratio $c<\frac{3}{2}$. Let the endpoints be $0$ and $M+2$, where $M$ is some large number.  Consider the case with $n=2\lceil{}\frac{M+1}{\epsilon}\rceil{}+4$ agents, where $1>\epsilon>0$. Consider the profile ${\bf x}$ which locates one agent at $1$ and $M+1$, locates an agent in each $1+a\epsilon \in (1,M+1)$ s.t. $a \in \mathbb{N}$, and splits the rest of the agents evenly among the two endpoints (if there is an odd number of agents remaining, locate one agent at $\frac{M+2}{2}$). Note that an optimal facility location is at $M+1\frac{1}{2}$, with a benefit of $\frac{1}{2}$. Let $p$ be the probability that the facility is located at $[1,M+1]$. It follows that the resulting expected benefit is upper bounded by $\frac{\epsilon}{2} p+\frac{1}{2}(1-p)$. To get the required approximation ratio, we must have $\frac{\epsilon}{2} p+\frac{1}{2}(1-p) \geq \frac{1}{2c}$, which gives $p \leq \frac{\frac{1}{2}(1-\frac{1}{c})}{\frac{1}{2}-\frac{\epsilon}{2}}$. The facility must be located either in $[0,1]$ or in $[M+1,M+2]$ with probability at least $\frac{1-p}{2}$.  Assume without loss of generality that it is located with probability at least $\frac{1-p}{2}$ in $[0,1]$. Consider the profile ${\bf x}'$, which is obtained from profile ${\bf x}$ by relocating the agents from $0$ so that there is an agent in every point $a\epsilon \in [0,1)$ s.t. $a \in \mathbb{N}$.  Let $p'$ be the probability that on this profile, the facility is located at $[0,M+1]$. Note that the optimal facility location remains $M+1\frac{1}{2}$ with benefit $\frac{1}{2}$, and on the other hand the expected benefit on this profile is bounded by $\frac{\epsilon}{2}p'+\frac{1}{2}(1-p')$, yielding the bound $p' \leq \frac{\frac{1}{2}(1-\frac{1}{c})}{\frac{1}{2}-\frac{\epsilon}{2}}$. Let us analyze the expected distance of the facility from $0$ in the two profiles. For ${\bf x}$, the expected distance from $0$ is no more than $\frac{1-p}{2}+p(M+1)+\frac{1-p}{2}(M+2)$. On the other hand, for ${\bf x}'$, the expected distance from $0$ is no less than $(1-p')(M+1)$. Since we have obtained ${\bf x}'$ from ${\bf x}$ using deviations of agents from $0$, strategyproofness dictates $\frac{1-p}{2}+p(M+1)+\frac{1-p}{2}(M+2) \geq (1-p')(M+1)$ \footnote{Consider the sequence of profiles ${\bf x^0}$ through ${\bf x^n}$, such that profile ${\bf x^i}$ agrees with ${\bf x'}$ on the location of agents $1$ through $i$ and with ${\bf x}$ on the location of the rest of the agents. Let $i^*$ be the index that maximizes $E_{Y \sim f({\bf x^i})}[Y]$; if there is more than one such index, choose the minimal one. If $i^*>0$, then it is beneficial for agent $i^*$ to deviate so that the profile changes from ${\bf x^{i^*-1}}$ to ${\bf x^{i^*}}$, violating strategyproofness. Thus $i^*=0$, implying the expected distance of the facility from $0$ in those profiles satisfies $E_{Y \sim f({\bf x^0})}[Y] \geq E_{Y \sim f({\bf x^n})}[Y]$. But, since ${\bf x^0}={\bf x}$ and ${\bf x^n}={\bf x}'$, we know that $\frac{1-p}{2}+p(M+1)+\frac{1-p}{2}(M+2) \geq E_{Y \sim f({\bf x})}[Y]$ and $E_{Y \sim f({\bf x'})}[Y] \geq (1-p')(M+1)$.}. Reorganizing this, we get: $\frac{3-p}{2}+\frac{1+p}{2}M \geq (1-p')M+1-p'$. Using our bounds for $p$ and $p'$, this implies the inequality $\frac{1}{2}+(\frac{1}{2}+\frac{\frac{1}{2}(1-\frac{1}{c})}{1-\epsilon})M \geq (1-\frac{\frac{1}{2}(1-\frac{1}{c})}{\frac{1}{2}-\frac{\epsilon}{2}})M-\frac{\frac{1}{2}(1-\frac{1}{c})}{\frac{1}{2}-\frac{\epsilon}{2}}$. Let us reorganize this inequality to $\frac{1}{2}+\frac{\frac{1}{2}(1-\frac{1}{c})}{\frac{1}{2}-\frac{\epsilon}{2}} \geq (\frac{1}{2}-\frac{\frac{3}{2}(1-\frac{1}{c})}{1-\epsilon})M$. As $c<\frac{3}{2}$, we can choose $\epsilon>0$ small enough so that $\frac{1}{2}-\frac{\frac{3}{2}(1-\frac{1}{c})}{1-\epsilon}>0$. Sending $M$ to $\infty$ then causes the r.h.s of the inequality to go to $\infty$, violating the inequality.
\end{proof}

\section{Mutiple Locations Per Agent in the Obnoxious Model}

In this section we follow the spirit of a suggestion in \cite{DBLP:journals/teco/ProcacciaT13} and study a generalized model, in which a single agent may be associated with more than one location. As this multiple location model is a generalization of our previous model, the lower bounds carry over; in particular, for the maxisum objective, we have lower bounds of $3$ and $\frac{2}{\sqrt{3}}$ on deterministic and randomized mechanisms respectively, even when $T=\{1\}$. We show that when $T=\{1\}$, in the deterministic case, we can find a strategyproof mechanism to match the lower bound, despite the additional power given to the agents. In addition, still when $T=\{1\}$, we provide a family of $\frac{3}{2}$- approximate randomized strategyproof mechanisms.

Our generalized model (for the definition of the model, we do not assume $T=\{1\}$) can be obtained from our previous model via the following changes. First, let ${\bf k}=(k_1, \ldots, k_n) \in \mathbb{N}^n$. A location profile is now ${\bf z}=({\bf z^1},{\bf z^2},\ldots,{\bf z^n})$, where for each $i=1,...,n$, ${\bf z^i}=(z^i_1,z^i_2,\ldots,z^i_{k_i}) \in I^{k_i}$. A deterministic mechanism is a collection of functions $f=\{f_n^{\bf k}:n \in \mathbb{N}, {\bf k} \in \mathbb{N}^n\}$, such that $f_n^{\bf k}:T^n \times I^{k_1} \times \ldots \times I^{k_n} \rightarrow I$ is a function that maps each pair of type and location profiles to a facility location. The benefit of agent $i$ from facility location $y$ is now defined as $B(\theta_i,{\bf z^i},y)=\sum_{j=1}^{k_i} B(\theta_i,z^i_j,y)$, where $B(\theta_i,x,y)$ is $|x-y|$ when $\theta_i=1$ and $2-|x-y|$ when $\theta_i=2$. The maxisum objective is $\sum_{i=1}^n B(\theta_i,{\bf z^i},y)$ as usual. The rest of the notation carries over, and the adjustment to the randomized model is easy and left to the reader. For the approximation ratio, we note that the possible instances of the problem include all possible options for both $n$ and ${\bf k}$.

First, we provide a $3$- approximate strategyproof deterministic mechanism for the case where $T=\{1\}$.

\begin{theorem}
Let $R^*=\{i: \frac{\sum_{j=1}^{k_i} z_j}{k_i} \leq 1\}$, $L^*=\{i:\frac{\sum_{j=1}^{k_i} z_j}{k_i}>1\}$. Let $f$ be the mechanism which locates the facility at $2$ if $\sum_{i \in R^*}k_i \geq \sum_{i \in L^*}k_i $ and at $0$ otherwise. This mechanism is strategyproof and $3$- approximate for the maxisum objective when $T=\{1\}$.
\end{theorem}

\begin{proof}
Strategyproofness is easy (note that $R^*$ is exactly the set of agents weakly preferring the facility to be located at $2$ over $0$). The optimal facility location is clearly in $\{0,2\}$. Assume without loss of generality that $f({\bf z})=2$. All we have to prove is that $\frac{sb({\bf z},0)}{sb({\bf z},2)} \leq 3$. But note that every agent $i \in R^*$ receives a benefit of at least $k_i$ when the facility is located at $2$ and at most $k_i$ when the facility is located at $0$. On the other hand, each agent $i \in L^*$ trivially gets a benefit between $0$ and $2k_i$. Thus, using the fact that $f({\bf z})=2$ implies $\sum_{i \in R^*}k_i \geq \sum_{i \in L^*}k_i$, we get $\frac{sb({\bf z},0)}{sb({\bf z},2)} \leq \frac{2\sum_{i \in L^*}k_i+\sum_{i \in R^*}k_i}{\sum_{i \in R^*}k_i} \leq \frac{2\sum_{i \in R^*}k_i+\sum_{i \in R^*}k_i}{\sum_{i \in R^*}k_i}=3$.
\end{proof}

Note that when $k_i=1$ for all $i$, this mechanism reduces to the mechanism proposed in \cite{DBLP:journals/tcs/ChengYZ13}.

Finally, we define a class of randomized strategyproof mechanisms that provide a $\frac{3}{2}$- approximation ratio when $T=\{1\}$ and show that it is nonempty.

\begin{theorem}
Let $f$ be a randomized mechanism that, for a profile ${\bf z}$, locates the facility at $0$ with probability $p_{\bf z}$ and at $2$ with probability $(1-p_{\bf z})$.  Then, when $T=\{1\}$, the following conditions on $p_{\bf z}$ are sufficient to make the mechanism strategyproof and $\frac{3}{2}$- approximate:
\begin{enumerate}
\item $p_{\bf z}$ is increasing in $\sum_{i \in L^*}k_i$ and decreasing in $\sum_{i \in R^*}k_i$.
\item $\frac{1}{3}+\frac{1}{6} \cdot \frac{\sum_{i \in L^*}k_i}{\sum_{i \in R^*}k_i} \geq p_{\bf z} \geq \frac{2}{3}-\frac{1}{6}\cdot\frac{\sum_{i \in R^*}k_i}{\sum_{i \in L^*}k_i}$ (if $\sum_{i \in R^*}k_i=0$, the leftmost term is $\infty$; if $\sum_{i \in L^*}k_i=0$, the rightmost term is $-\infty$).
\end{enumerate}
Furthermore, the class of mechanisms of this form is nonempty.
\end{theorem}

\begin{proof}
Strategyproofness is clear. Fix ${\bf z}$, and set $p=p_{\bf z}$. For the approximation ratio, there are two cases to consider. First, assume that the optimal facility location for profile ${\bf z}$ is $0$, with social benefit $OPT$. If $0$ and $2$ are both optimal, clearly any choice of $p$ yields approximation ratio $1$. Assume $2$ is not optimal; then $\sum_{i \in L^*}k_i>0$. As for every $i \in R^*$ we have that $\frac{\sum_{j=1}^{k_i} z_j}{k_i} \leq 1$, his benefit from locating the facility at $2$ is at least $k_i$, and so the social benefit from locating the facility there is at least $\sum_{i \in R^*}k_i$. Thus, it is enough to prove that $pOPT+(1-p)\sum_{i \in R^*}k_i \geq \frac{2}{3}OPT$, or equivalently $(1-p)\sum_{i \in R^*}k_i \geq (\frac{2}{3}-p)OPT$. If the right hand side is negative, then this inequality is satisfied. Assume that the right hand side is nonnegative. Note that $OPT \leq 2\sum_{i \in L^*}k_i+\sum_{i \in R^*}k_i$ (the benefit of $i \in R^*$ from locating the facility at $0$ is bounded by $k_i$, while the benefit of $i \in L^*$ is trivially bounded by $2k_i$). Thus, it is enough to prove that $(1-p)\sum_{i \in R^*}k_i \geq (\frac{2}{3}-p)(2\sum_{i \in L^*}k_i+\sum_{i \in R^*}k_i)$. Isolating $p$ in this inequality gives $p \geq \frac{2}{3}-\frac{1}{6}\frac{\sum_{i \in R^*}k_i}{\sum_{i \in L^*}k_i}$, which is satisfied.\\

On the other hand, assume that the optimal facility location for profile ${\bf z}$ is $2$; note that this implies $\sum_{i \in R^*}k_i>0$. Similarly to the analysis above, we can get a lower bound of $\sum_{i \in L^*}k_i$ on the benefit of locating the facility at $0$ and upper bound of $2\sum_{i \in R^*}k_i+\sum_{i \in L^*}k_i$ on $OPT$. Thus, we need that $(1-p)OPT + p\sum_{i \in L^*}k_i \geq \frac{2}{3}OPT$,  and so it is enough to verify that $p\sum_{i \in L^*}k_i \geq (p-\frac{1}{3})(2\sum_{i \in R^*}k_i+\sum_{i \in L^*}k_i)$. Isolating $p$ yields $p \leq \frac{1}{3}+\frac{1}{6}\frac{\sum_{i \in L^*}k_i}{\sum_{i \in R^*}k_i}$.\\

Finally, we verify that $p=\max{\{\frac{2}{3}-\frac{1}{6}\cdot\frac{\sum_{i \in R^*}k_i}{\sum_{i \in L^*}k_i},0\}}$ (where if $\sum_{i \in L^*}k_i=0$, $p=0$) satisfies the above properties. The only thing that requires proof is $p \leq \frac{1}{3}+\frac{1}{6} \cdot \frac{\sum_{i \in L^*}k_i}{\sum_{i \in R^*}k_i}$ (assuming $\sum_{i \in R^*}k_i>0$; if $\sum_{i \in R^*}k_i=0$ then there is nothing to prove). Note that the right hand side is positive, so it is enough to show is that $\frac{1}{3}+\frac{1}{6} \cdot \frac{\sum_{i \in L^*}k_i}{\sum_{i \in R^*}k_i} \geq \frac{2}{3}-\frac{1}{6}\cdot\frac{\sum_{i \in R^*}k_i}{\sum_{i \in L^*}k_i}$ when $\sum_{i \in L^*}k_i>0$. But this is equivalent to $\frac{(\sum_{i \in L^*}k_i)^2+(\sum_{i \in R^*}k_i)^2}{(\sum_{i \in L^*}k_i)(\sum_{i \in R^*}k_i)} \geq 2$, and note that $\frac{(\sum_{i \in L^*}k_i)^2+(\sum_{i \in R^*}k_i)^2}{(\sum_{i \in L^*}k_i)(\sum_{i \in R^*}k_i)}=2+\frac{(\sum_{i \in L^*}k_i-\sum_{i \in R^*}k_i)^2}{(\sum_{i \in L^*}k_i)(\sum_{i \in R^*}k_i)} \geq 2$.
\end{proof}

It is worth noting that the randomized mechanism given in \cite{DBLP:journals/tcs/ChengYZ13}, for the special case of $k_i=1$ for all $i$, falls into the category of mechanisms we defined here.

\section{Future Research}

There are many additional possible directions for this research. The immediate question stemming from our results is what further improvement can be achieved in the approximation ratio under the hybrid model by using randomization. Another interesting question is whether it is possible to derive a clear characterization of randomized strategyproof mechanisms. Other directions include characterization and bounds for topologies different than the interval, and for objectives other than maxisum and egalitarian.

\bibliography{BIBLIOGRAPHY}
\bibliographystyle{plain}

\end{document}